\documentclass[aps]{revtex4}
\usepackage{amsmath}
\usepackage{amssymb, amsmath}
\usepackage{amsmath,calc}
\usepackage{amsthm}
\usepackage{amsfonts}
\usepackage{graphicx}
\usepackage{color}
\usepackage{float}
\usepackage{braket}
\usepackage{algorithm}
\usepackage{algpseudocode}
\usepackage{multirow}
\usepackage{listings}
\usepackage{endnotes}

\graphicspath{{imgs/}} 

\begin{document}
\title{An almost-solvable model of complex network dynamics}

\author{Qi Guo$^1$, Artur Sowa$^{1}$}
\affiliation{$^{1}$Department of Mathematics and Statistics, University of Saskatchewan}

\newtheorem{definition}{Definition}[section]
\newtheorem{theorem}{Theorem}[section]
\newtheorem{lemma}{Lemma}[section]
\newtheorem{corollary}{Corollary}[section]
\newtheorem{remark}{Remark}
\newtheorem{conclusion}{Conclusion}
\newtheorem{example}{Example}
\newtheorem{proposition}{Proposition}[section]

\begin{abstract}

We discuss a specific model, which we refer to as RandLOE, of a large multi-agent network whose dynamic is prescribed via a combination of deterministic local laws and random exogenous factors. The RandLOE approach lies outside the framework of Stochastic Differential Equations, but lends itself to analytic examination as well as to stable simulation even for relatively large networks.

RandLOE is based on the logistic operator equation (LOE), which is a multidimensional dynamical system extending the classical logistic equation via an operator-algebraic interaction term.   The network is defined by interpreting the LOE variable as an adjacency matrix of a complete graph. Depending on the choice of parameters, it can display a number of essentially distinct dynamical characteristics: e.g. cycles of expansion and contraction.



\vspace{.5cm}

\vspace{.5cm}

\vspace{.5cm}

\end{abstract}
\date{\today}
\maketitle

\section{Introduction}

Complex Network Theory (CNT)
aspires to furnish a unifying quantitative foundation for a number of scientific disciplines, including  Economics, Computer Science, Biology, and Social Sciences, (see e.g. \cite{Souma}, \cite{Albert}, \cite{Thiery}, \cite{Girvan}, or a review in \cite{Newman}).
While there is no generally accepted standard definition of CNT, the following characterisation seems to capture much of the essence: ``[CNT is a study of] networks whose structure is irregular, complex and dynamically evolving in time, with the main focus moving from the analysis of small networks to that of systems with thousands or millions of nodes, and with a renewed attention to the properties of networks of dynamical units", \cite{Boccaletti}. The emancipation of CNT from its graph-theoretic roots is most manifest in the character of new results, such as the discovery of a power-law governing random, large networks, such as the www, \cite{Faloutsos}.
Another distinctive feature of CNT is its emphasis on statistical properties as well as on various questions of control, stability, robustness, \cite{Albert},  \cite{Tannenbaum}, sensitivity of dynamics to network structure, \cite{Nishikawa}, and on questions of epidemiological character, \cite{Pastor}. New types of structures have been found to occur in real networks, such as the small-world structure, \cite{Watts}, \cite{Amaral}. These discoveries shed some light on the function of most complex systems, such as the human brain, \cite{Liu}, \cite{Wu}.

Our focus in this work is on quantitative analysis of network dynamics. There is a body of preexisting work in this direction, e.g. results  on small-world networks of oscillators, \cite{Hong}, or scale-free random networks, \cite{Barabasi}.  An alternative approach is via operator equations, particularly the logistic operator equation (LOE), \cite{Sowa2015interacting}. While inspired by considerations from Quantum Statistical Mechanics, the LOE, and the related stochastic process, may be interpreted as prescribing complex network dynamics. Indeed, the LOE captures complex interactions in a hierarchical structure that is either pulled toward equilibrium or forced into a cyclical shift between expansion and contraction, while in both cases being perturbed by random exogenous forces. While, in principle, similar models could be formulated within the frame of Stochastic Differential Equations, computational simulation of such structures is prohibitively slow and inaccurate for a large number of variables. In contrast, the LOE is amenable to accurate simulation even for a relatively large size. That is due to the fact that the solutions are in large part determined by explicit, perfectly accurate, recursive schemas.

Our focus in this work is on new types of solutions of the LOE and the stochastic processes that they induce.  One of the main observations is that a special (unimodal, i.e. non-hierarchical) case of such a process is, in fact, an It\^{o} process, see Section \ref{Section_Motivation}. In Section \ref{Section_LOE_deterministic} we study a much more complex multimodal LOE, and provide an algorithmic description of the deterministic solutions. While in the seminal work, \cite{Sowa2015interacting}, the multimodal hierarchy of modes (eigenvalues) scales as $1/\log p$ where $p$ are consecutive primes, in the present work it scales as $1/m$, where $m$ are consecutive integers. In Section \ref{Section_RandLOE} we outline some properties of the resulting stochastic process. We summarize in Section \ref{Section_Summary} by highlighting the network interpretation. The results presented here are mostly theoretical. However, many insights, not to mention the showcased graphs, grew out of numerical experiments, involving custom codes in MATLAB, also using a third-party package, \cite{NCSOStools}, Python, \cite{Python}, and  Cytoscape, \cite{Cytoscape}. Although this report is focused on theoretical aspects, we envision a number of real-life applications. In particular, this framework may be deployed toward valuation of futures associated with commercial networks that display significant interactivity and hierarchization. It may also be used toward analyses of the stability of such networks, opening toward applications in macroeconomics. The importance of complex networks approach to stability analysis has been highlighted in \cite{Battiston}.

\section{The resolvent and the unimodal RandLOE}\label{Section_Motivation}

We fix a Hilbert space $\mathbb{H}$ (generally complex, but alternatively real when so indicated)  and an a basis $\ket{e_n}$ ($n= 1,2\ldots N$).  Let $X: \mathbb{H}\rightarrow \mathbb{H}$ be an operator.
Recall, the resolvent of $X$ is defined\footnote{Note that another variant of the definition of the resolvent, namely $R(z)=(zI-X)^{-1}$, is also encountered in literature. } as
\begin{equation}\label{Def_resolvent}
  R = R_X(z) = (I-zX)^{-1}, \quad z\in \mathbb{C}
\end{equation}
Let $\lambda_n$ ($n=1,\ldots N$) be the eigenvalues of $X$. It is well known, see \cite{Kato}, that $z\mapsto R(z)$ is a meromorphic matrix-valued function. Its only poles occur at points $z= 1/\lambda_n$. Locally, say, around $z=0$, the resolvent may be represented via the power series
$
  R = \sum_{m=0}^{\infty} X^m z^m.
$
It seems less commonly known that the resolvent satisfies the following differential equation:
\begin{equation}\label{Oper_Eq_resolvent}
-z \frac{d}{dz} R = R - R^2.
\end{equation}
Indeed, applying the product rule to $\frac{d}{dz} [RR^{-1}] = 0,$ one readily finds $\frac{d}{dz} R = RXR.$ On the other hand, $R-R^2 = R(R^{-1} - I)R = -zRX R$, hence (\ref{Oper_Eq_resolvent}). Equation (\ref{Oper_Eq_resolvent}) and its generalizations are the main focus of this work. We observe that for a smooth curve  $t\mapsto z(t)$, where $t$ is a real parameter interpreted as time, (\ref{Oper_Eq_resolvent}) yields
\begin{equation}\label{path_LOE_specia}
\frac{d}{dt} R = -\frac{z'}{z}(R - R^2).
\end{equation}
This equation furnishes a model for deterministic dynamic. It also induces a stochastic dynamic as follows:
Assume that $\mathbb{H} = \mathbb{R}^N$ is a real Hilbert space, and let
\begin{equation}\label{X_matrix}
 X(t)= \sum_{i,j=1}^{N}W_{ij}(t) \ket{e_i} \bra{e_j}
\end{equation}
be a random matrix whose entries are mutually independent standard Wiener processes $W_{ij}(t)$. We define a special (unimodal) case of the RandLOE as
\begin{equation}\label{Def_randloe_special}
  t \mapsto R(t)= R_{X(t)}(z(t)) = [\,I - z(t) X(t)\,]^{-1}.
\end{equation}
Next, we demonstrate that this is an It\^{o} process. The proof requires the following observation:
\begin{lemma}\label{special_case_sde_lemma} 	
	For an arbitrary $N$-by-$N$ matrix $Y$ and $X(t)$ as in (\ref{X_matrix}), we have
	\begin{equation}\label{lemma_special case}
	dX \, Y \, dX = Y^{T} dt,
	\end{equation}
	where $Y^T$ denotes the transpose of $Y$.
\end{lemma}

\begin{proof}
	 Let $Y= \sum_{i,j=1}^{N} Y_{ij} \ket{e_i} \bra{e_j}$, where $Y_{ij}$ are coefficients of $Y$. We have
	
	\begin{equation}\label{lemma_special case_proof}
	\begin{split}
	dX\, Y\, dX &=  (\sum_{i,j=1}^{N} dW_{ij} \ket{e_i} \bra{e_j}) (\sum_{i,j=1}^{N} Y_{ij} \ket{e_i} \bra{e_j}) (\sum_{i,j=1}^{N} dW_{ij} \ket{e_i} \bra{e_j}) \\& = \sum_{i,j=1}^{N} \sum_{k,l=1}^{N} dW_{ik}dW_{lj}Y_{kl}  \ket{e_i} \bra{e_j} = \sum_{i,j=1}^{N} Y_{ji} dt \ket{e_i} \bra{e_j}  = Y^{T} dt,
	\end{split}
	\end{equation}
	where we have used identity $dW_{ik}dW_{lj}=\delta_{il}\delta_{kj}dt$.	
\end{proof}

The following is the main result of this section:

\begin{theorem}\label{resolvent_sde_theorem}
	$R(t)$ defined in (\ref{Def_randloe_special}) is an It\^{o} process, and satisfies
	\begin{equation}\label{resolvent_SDE}
	dR=\left[-\frac{z'}{z}( R - R^2)+ z^2\,RR^{T}R\right]\,dt+z\,R \, dX \, R.
	\end{equation}
\end{theorem}

\begin{proof} 
 Note that the It\^{o} lemma indicates the general fact that $R$ satisfies a stochastic differential equation; we only need to identify its form.  To this end, we engage infinitesimal stochastic calculus. First, observe
\[
dR^{-1} = d[I - z X] = - z' X dt - z dX,
\]
and
\[
0 = d(RR^{-1}) = dR\, R^{-1} + R\, dR^{-1} +  dR\, dR^{-1}.
\]
Combining these identities we obtain
\[
dRR^{-1} = z'R\,Xdt + zR\,dX + dR\, (z' X dt + z dX) = z'R\,Xdt + z\,R\,dX + z\, dR\, dX,
\]
where we have cancelled the term $z' dR\,Xdt$ which has order $O(t^{3/2})$.
Next, observing
$
R - R^2 = R(R^{-1}-I)R = -z RXR,
$
we obtain
 \[
dR = -\frac{z'}{z}(R-R^2)\,dt + z\,R\,dX\,R + z\, dR\, dX\,R.
\]
At this stage we use a bootstrapping argument. Namely, the term $z\, dR\, dX\,R$ has only the $dt$ component, and that can only depend  on the $dX$ component of $dR$. Since the latter is $z\,R\,dX\,R$, we obtain
 \[
dR = -\frac{z'}{z}(R-R^2)\,dt + z\,R\,dX\,R + z^2 R\,dX\,R\, dX\,R.
\]
In light of (\ref{lemma_special case}) this is equivalent to (\ref{resolvent_SDE}).
\end{proof}
Stochastic process (\ref{Def_randloe_special}) is well defined as long as $1/z(t)$ does not coincide with any of the eigenvalues of $X(t)$. This may be ensured for a considerable time interval $t\in [0,T]$ by choosing the initial conditions in which $1/z(0)$ is separated from all the eigenvalues of $X(0)$ (which are all zero for the standard Wiener process) by considerable distance. However, other considerations may also affect the choice of curve $z(t)$. In particular, it is especially interesting to examine the paths $z=\exp (it)$ and $z=\exp (-t)$. In the first case the resolvent series may be interpreted as a Fourier series, and in the second as a generalized Dirichlet series. The first type of process is adequate to modelling phenomena that display some cyclicity, whereas the second type to phenomena that display damping in some charcteristic time period.

\section{LOE in a new regime}  \label{Section_LOE_deterministic}

The special case of LOE given by (\ref{Oper_Eq_resolvent}) is characterized by \emph{unimodality}, i.e. the diagonal solutions are scalar multiples of the identity; namely,  $(1-az)^{-1} I$ for some parameter $a$. We now turn attention to a \emph{multimodal} LOE, which admits more complex diagonal solutions.

As before, we fix a Hilbert space $\mathbb{H}$ with a distinguished basis $\ket{e_m}$. In this section, it may be either finite dimensional ($m = 1,2, \ldots, N$) or infinite-dimensional ($N=\infty$), as well as complex or real.  The multimodal LOE is defined as:
\begin{equation}\label{LOE_general_z}
-z\, \Lambda \frac{d}{dz} F= F-F^2,\quad \mbox{ where } \Lambda = \sum_{m} \frac{1}{m} \ket{e_m}\bra{e_m}.
\end{equation}
The dependent variable $z\mapsto F(z)$ is an analytic operator-valued function with $F(z):\mathbb{H}\rightarrow\mathbb{H}$. An elementary argument shows that the nontrivial\footnote{We do not consider the trivial vanishing solutions.} diagonal solution $F(z)$ is necessarily of the form
\[
F(z) = \sum_{m=0}^{\infty} \frac{1}{1-a_mz^m}\ket{e_m}\bra{e_m},
\]
where $a_m$ are arbitrary. The diagonal entries may be interpreted as eigenmodes of the system described by the dynamic (\ref{LOE_general_z}).

Next, we wish to consider analytic solutions of (\ref{LOE_general_z}) in full generality. In contrast to (\ref{Oper_Eq_resolvent}), we are not aware of the closed-form formula for solutions of this equation. Thus, we resort to a search for solutions in the form of a power series with matrix coefficients. The outcome is a recurrence formula for the coefficients, which characterizes such solutions, namely:

\begin{theorem} \label{Theorem_recurrence}
Assume $F = \sum_{m=0}^{\infty} F_m z^m$ satisfies  (\ref{LOE_general_z}) and $F_0=I$. Then  $F_1=\ket{e_1} \bra{v_1}$ for an arbitrary vector $\ket{v_1} \in \mathbb{H} $. Moreover:
	\begin{itemize}
		
		\item
		When $\mathbb{H} = \mbox{span} \{ \ket{e_k}:\, k \in \mathbb{N}  \}$ (infinite-dimensional Hilbert space) we have for $m>1$
		\begin{equation}\label{LOE_recurrence}
F_m(t)=Q_m\sum_{k=1}^{m-1} F_k F_{m-k}+\ket{e_m} \bra{v_m}, \quad \mbox{ where }\quad Q_m= \sum_{k=1}^{m-1}\frac{k}{m-k} \ket{e_k} \bra{e_k},
		\end{equation}
		and vectors $\ket{v_{m}}$ are arbitrary.
		
		\item
		When $\mathbb{H} = \mbox{span} \{ \ket{e_k}:\, k =1,2,\ldots N  \}$, formula (\ref{LOE_recurrence}) remains valid with the convention $v_m=0$ for all $m>N$.
	\end{itemize}
\end{theorem}

\begin{proof} First, we assume that the Hilbert space is infinite dimensional.
	Substituting the power series of $F$  into (\ref{LOE_general}) and comparing coefficients we obtain
	\begin{equation}\label{rec_bas}
	-m \Lambda  F_m = F_m - \sum_{k=0}^{m} F_k F_{m-k}.
	\end{equation}
	For $m=0$, this implies $0=F_{0}-F_{0}^{2}$, i.e. $F_0$ need to be idempotent; in particular, $F_0=I$, which we chose \emph{a priori}, is admissible. Next,  (\ref{rec_bas}) yields	
	\[
	(I-\Lambda) F_1= 0,
	\]
Note that  $ Ker \, (I-\Lambda)= \mbox{span} \{ \ket{e_1} \}$. In particular, the operator $I-\Lambda$ is invertible when restricted to  the orthogonal complement of $\mbox{span} \{ \ket{e_1} \}$. This readily implies
			\begin{equation}\label{LOE_first}
			 	F_1=\ket{e_1} \bra{v_1} \quad 	\mbox{ with  arbitrary } \ket{v_1}.
			\end{equation}
 Next, when $m>1$, (\ref{rec_bas}) yields
	\begin{equation}\label{rec_bas_2}
	(I-m\Lambda) F_m=- \sum_{k=1}^{m-1} F_k F_{m-k},
	\end{equation}
As before, we notice that  $ Ker \, (I-m\Lambda)= \mbox{span} \{ \ket{e_m} \}$. In particular, the operator $I-m\Lambda$ is invertible when restricted to  the orthogonal complement of $\mbox{span} \{ \ket{e_m} \}$. Moreover, the inverse of the restricted operator is given explicitly as
\[
\sum_{\substack{k=1 \\ k\neq m}}^{\infty} (1-\frac{m}{k})^{-1} \ket{e_k} \bra{e_k} =
\sum_{\substack{k=1 \\ k\neq m}}^{\infty} \frac{k}{k-m} \ket{e_k} \bra{e_k}
\]
Thus, $F_m$ is determined by all $F_k$ with $k<n$, but only up to a term $\ket{e_m} \bra{v_m}$ for an arbitrary vector $\ket{v_m}$. Namely, from (\ref{rec_bas_2}) we obtain a recurrence formula in the form
	\[
	F_m= \left(\sum_{\substack{k=1 \\ k\neq m}}^{\infty} \frac{k}{m-k} \ket{e_k} \bra{e_k}\right)\, \sum_{\substack{k=1}}^{m-1} F_k F_{m-k}+\ket{e_m} \bra{v_m}.
	\]
It follows from this formula and (\ref{LOE_first}) by induction that the range of $F_m$ is a subspace of $\mbox{span} \{ \ket{e_k}:\, k \leq m  \}$. Hence the formula simplifies to the form
(\ref{LOE_recurrence}).

Finally, when the Hilbert space has finite dimension $N$, the operators $I - m\Lambda$ are invertible (in the entire space) for all $m>N$. Thus, the above argument leading to the recurrence formula may be repeated verbatim, except that terms $\ket{e_m}\bra{v_m}$ need to be set to zero when $m>N$. This completes the proof.	
\end{proof}
\begin{remark}\label{Example_polynomials}
Denote $b_n = \ket{e_n}\bra{v_n}$. The recurrence formula indicates that each $F_n$ is a polynomial in variables $b_n$. Here are a few examples:

	\[
	F_1=b_1
	\]

	\[
	F_2=b_1^2 + b_2
	\]

	\[
	F_3=b_1^3+\frac{1}{2}b_1b_2+2b_2b_1+b_3
	\]

	\[
	F_4=b_1^4+\frac{1}{2}b_1^2b_2+\frac{5}{6}b_1b_2b_1+\frac{1}{3}b_1b_3+3b_2b_1^2+3b_3b_1+b_2^2+b_4
	\]

	\[
	\begin{split} F_5=&b_1^5+\frac{1}{2}b_1^3b_2+\frac{5}{6}b_1^2b_2b_1+\frac{1}{3}b_1^2b_3+\frac{13}{12}b_1b_2b_1^2+\frac{5}{6}b_1b_3b_1+\frac{3}{8}b_1b_2^2+\frac{1}{4}b_1b_4+\\&+4b_2b_1^3+6b_3b_1^2+2b_2^2b_1+\frac{3}{2}b_3b_2+\frac{2}{3}b_2b_3+ \frac{5}{3}b_2b_1b_2+b_5+4b_4b_1
	\end{split}
	\]

	\[
	\begin{split} F_6&=b_1^6+\frac{1}{2}b_1^4b_2+\frac{5}{6}b_1^3b_2b_1+\frac{1}{3}b_1^3b_3+\frac{13}{12}b_1^2b_2b_1^2+
\frac{3}{8}b_1^2b_2^2+\frac{5}{6}b_1^2b_3b_1+\frac{1}{4}b_1^2b_4+\frac{77}{60}b_1b_2b_1^3+\\&+\frac{11}{20}b_1b_2b_1b_2+\frac{27}{40}b_1b_2^2b_1+
\frac{7}{30}b_1b_2b_3+\frac{43}{30}b_1b_3b_1^2+\frac{11}{30}b_1b_3b_2+\frac{17}{20}b_1b_4b_1+\frac{1}{5}b_1b_5+\\&+5b_2b_1^4+\frac{9}{4}b_2b_1^2b_2+
\frac{13}{4}b_2b_1b_2b_1+\frac{7}{6}b_2b_1b_3+3b_2^2b_1^2+b_2^3+\frac{11}{6}b_2b_3b_1+\frac{1}{2}b_2b_4+\\&+10b_3b_1^3+\frac{7}{2}b_3b_1b_2+
\frac{7}{2}b_3b_2b_1+b_3^2+10b_4b_1^2+2b_4b_2+5b_5b_1+b_6
	\end{split}
	\]

The general formula is not given explicitly. However, we can summarize as follows:
	\begin{equation}\label{LOE_mean_diagonal}
	F_m=\sum\limits_{i_1+\cdots +i_p=m} c_{i_1,i_2,\cdots,i_p}\,b_{i_1}b_{i_2}\cdots b_{i_p},
	\end{equation}
	where $c_{i_1,i_2,\cdots,i_p}$ are some constant coefficients. (We use the convention  $b_{0}=I$.) A direct induction argument based on (\ref{LOE_recurrence}) yields $c_{1,1,\cdots,1} = 1$, so that the right hand side of (\ref{LOE_mean_diagonal}) always contains the term  $b_1^m$.

We also note that  the recurrence  yields an upper bound on the growth of the operator norm of $F_n$. Indeed, denoting $x_n = \|b_n\| = \|v_n\|$, and using the subadditivity and submultiplicativity of the operator norm, we see that $\|F_n\|$ is bounded above by a polynomial in $x_1, x_2, \ldots x_n$, e.g. one readily finds
$
\|F_3\| \leq 4x_1^3 + 4x_1x_2 + x_3.
$

We emphasize that we do not undertake here the problem of analytic continuation, or even convergence of the series $F(z)$, cf. Remark \ref{Remark_conv}.
\end{remark}

\section{The multimodal RandLOE} \label{Section_RandLOE}

In what follows
 we only consider the finite-dimensional $N$-by-$N$ version of the LOE. We start by observing that choosing a path in the complex plane $t\mapsto z(t)$ gives a deterministic dynamic equation for $t\mapsto F(z(t))$, namely
\begin{equation}\label{LOE_general}
-\Lambda \frac{d}{dt} F= \frac{z'}{z} (F-F^2).
\end{equation}
We write $F(z(t))=F(t)$ for short.
Note that for the special choice $z= e^{-t}$ the diagonal modes are logistic curves and for $z = \exp it$, complex-valued logistic curves. In either case, the non-diagonal solutions are more interesting.
The recurrence formula given in Theorem \ref{Theorem_recurrence} enables simulation of such solutions with nearly perfect accuracy.  Recall, the singular value decomposition of an arbitrary matrix $F(t) \in \mathbb{C}^{N \times N}$, namely:
\[ F(t)= U(t) \, S(t) \, V(t)  \quad \mbox{where} \quad S(t)=\sum_{j=1}^{N} s_j(t) \ket{e_j}\bra{e_j},  \]
and both $U(t)$ and $V(t)$ are unitary.
The singular values $s_j(t)$ are the square
roots of the eigenvalues of $F(t)^{\dagger} F(t)$, where $F(t)^\dagger$ is the adjoint of $F(t)$.  Fig. \ref{Fig_Sing_Values} illustrates the evolution of singular values of $F(t)$ for the two highlighted curves.

\begin{remark}\label{Remark_conv}
The numerical work relies on partial sums of the series for $F(z(t))$. Heuristic arguments and numerical experiments strongly point to series convergence when $z(t) = e^{-(\alpha +i\beta) t }$ (for $t>\varepsilon> 0$) with an arbitrary $\beta $ and positive $\alpha$. We point out that the right-hand side of (\ref{LOE_general}) is only locally Lipschitz continuous. Thus, based on ODE theory,  solutions of the initial value problem are only guaranteed to exist locally in time; in other words, the possibility of finite time blowup is not excluded \emph{a priori}. Matters are even more complicated as regards the series solutions that are our focus. Indeed, we have $F(0) = I+ \sum_{n=1}^{\infty} F_n$, which series is not \emph{a priori} guaranteed to converge. Such convergence conundrums are characteristic of combinatorially complex symmetry-based calculations, Quantum Field Theory being another example.

\end{remark}

\begin{figure}[H]
	\centering
	\includegraphics[width=100mm]{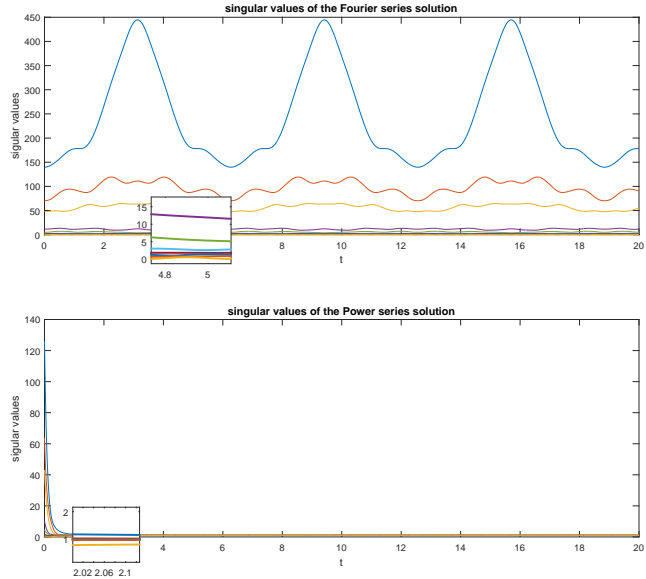}
	\caption{The singular values of $F(t)$, $N=10$ when $z=\exp (it)$ (top) and $z = \exp (-t)$ (bottom).}
	\label{Fig_Sing_Values}
\end{figure}

We wish to consider a LOE-based stochastic process obtained by randomizing the exogenous variables in (\ref{LOE_recurrence}).  Namely, let the exogenous variables $v_m = v_m(t)$ be time-dependent in the form of mutually independent $N$-dimensional Wiener processes, i.e.
\begin{equation}\label{Wiener}
\ket{v_m(t)} = \sum_{j=1}^{N} W_{mj}(t)\, \ket{e_j}, \quad k=1,2,\ldots ,N,
\end{equation}
where $W_{mj}(t)$ are standard mutually independent Wiener processes.
Formula (\ref{LOE_recurrence}) renders $F_n=F_n(t)$, $n \in \mathbb{N}$, as time-dependent stochastic processes. Subsequently, this gives rise to the main stochastic process:
\begin{equation}\label{RandLOE}
t\mapsto F(t; v_1(t), v_2(t),\ldots v_N(t))=\sum_{n=0}^{\infty} F_n(t) e^{-z nt}.
\end{equation}
We will refer to it as multimodal RandLOE. This random process is easy to simulate numerically by combining a simulation of (\ref{Wiener}), e.g. via the Karhunen-Lo\`{e}ve method,  with computation of $F_n(t)$ via (\ref{LOE_recurrence}). An example of the evolution of the singular values of $F(t)$  for the two highlighted curves is given in Fig. \ref{Fig_RandLOE_KL}.

\begin{figure}[H]
	\centering
	\includegraphics[width=100mm]{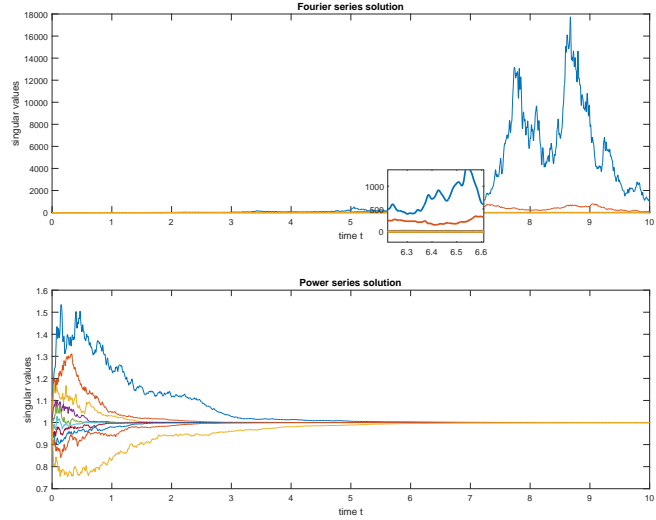}
	\caption{Singular values of the LOE-based stochastic process. All $v_m(t)$ are generated by Karhunen-Lo\`{e}ve method; $N=10$. The top graph shows that the dynamic corresponding to $z(t) = \exp (it)$ leads to intermittent expansions and contractions, whereas the dynamic based on $z(t) = \exp (-t)$ is dampened after some characteristic time period.}
	\label{Fig_RandLOE_KL}
\end{figure}

We would like to know if the multimodal RandLOE is also an It\^{o} process, just as its unimodal counterpart (\ref{Def_randloe_special})?  What differentiates it from the special case is lack of a closed-form formula for $F(t)$. This may be more than a mere technicality. Indeed, the multimodal RandLOE seems to have a greater inherent complexity. Based on this observation, we venture a hypothesis that the general RandLOE is not an It\^{o} process. At present this is an open problem.

In order to gain some additional insight into the nature of multimodal RandLOE we investigate its mean. This is  somewhat akin to the mean field approach to network analysis, \cite{Barabasi}.
We observe the following:

\begin{proposition} Let $F_m(t)$ be the coefficients of multimodal RandLOE.
When $m$ is odd $E[F_m(t)]=0$. For an even $m$, $E[F_m(t)]$ is a diagonal matrix whose entries are polynomials in $t$ of degree  $m/2$.
\end{proposition}

\begin{proof}
The proof relies essentially on Theorem \ref{Theorem_recurrence}.
As before, let $\ket{e_m} \bra{v_m(t)}=b_m(t)$ where $v_m$ are as in (\ref{Wiener}).
We invoke  the well-known formula for the moments of a normally distributed variable: 	
	\[
	E[W(t)^r]=\begin{cases}
	t^{r/2}(r-1)!\, ! & r \mbox{ even} \\
	0 & r \mbox{ odd. }
	\end{cases}
	\]
It follows that the only polynomials in $b_k$ that have a nonzero expected value are periodic with an even number of cycles. More precisely, these are of the form
\[
\left(b_{i_1}b_{i_2}\ldots b_{i_n}\right)^r \quad (n\geq 1,  r \mbox{ even}).
\]
We readily obtain
\begin{equation}\label{expect_cycle}
  E\left[\left(b_{i_1}b_{i_2}\ldots b_{i_n}\right)^r\right] =  \left[ t^{r/2}(r-1)!\, ! \right]^n \ket{e_{i_1}}\bra{e_{i_1}}\quad (r \mbox{ even}).
\end{equation}
Note that if $r$ is even, then so is $m$. We emphasize that all other polynomials have vanishing expectation.
On the other hand, (\ref{LOE_recurrence}) implies
\[
E[F_m(t)] = Q_m\sum_{k=1}^{m-1} E[\,F_k(t) F_{m-k}(t)\,].
\]
It follows by induction that $F_k(t) F_{m-k}(t)$ is a polynomial of degree $m$ in $b_k(t)$ with $k \leq m-1$, cf. (\ref{LOE_mean_diagonal}). (\ref{expect_cycle}) implies in particular that $E[F_m(t)]$ is a diagonal matrix. In addition, $E[F_m(t)]\neq 0$ only if $m$ is even. An inductive argument also shows that $F_m(t)$ contains the term $b_1(t)^m$, which contributes a monomial of the highest possible degree $m/2$. This completes the proof.
\end{proof}

\begin{example}
Proceeding as in Remark \ref{Example_polynomials} we present some computed examples\footnote{We have used a Matlab package NCSOStools and a custom Python code to accomplish this computation.}:
\[
E[F_2(t)] =  \mbox{diag}(t,0,  \ldots, 0)
\]
	\[
	E[F_4(t)]=\mbox{diag}(3t^2,t,0,\ldots,0)
	\]
	\[
	E[F_6(t)]=\mbox{diag}(15t^3+\frac{11}{20}t^2,\frac{13}{4}t^2,t,0,\ldots,0)
	\]
	\[
E[F_8(t)]=\mbox{diag}(105t^4+\frac{31}{63}t^2+\frac{1723}{315}t^3,3t^2+\frac{268}{15}t^3,\frac{37}{5}t^2,t,0,\ldots,0)
\]
\[
E[F_{10}(t)]=\mbox{diag}(945t^5+\frac{4759}{84}t^4+\frac{2866}{303}t^3+\frac{23}{48}t^2,\frac{7599}{56}t^4+\frac{70589}{4032}t^3+\frac{47}{72}t^2,\frac{482}{7}t^3+\frac{53}{28}t^2,\frac{27}{2}t^2,t,0,\ldots,0)
\]
\end{example}

\begin{remark} It is interesting to observe the analogous result for the unimodal RandLOE. Recall that the series coefficients in this case are $X^m$, where $X=X(t)$ is an $N$-by-$N$ matrix whose entries are independent standard Wiener processes. It is easy to observe that
$E[X^m] = C(m)\, t^{m/2} I$ for $m$ even, where $C= C(m)$ is a constant. Also, $E[X^m]=0$ when $m$ is odd. Thus, the expected values of coefficients in the unimodal case are monomials, contrasting with the nontrivial polynomials of the multimodal case. To our knowledge this by itself does not provide any clues as to whether the latter case is an It\^{o} process or not.
\end{remark}

\section{Summary: RandLOE as a complex network dynamics} \label{Section_Summary}

A network is understood to be a complete graph. Its vertices represent the nodes, e.g. market players, internet servers, etc. Any pair of nodes are connected by two edges (arrows) with opposite orientations. The arrows represent exchange channels for transporting goods or information. In addition, each node is connected to itself via an un-oriented looping edge. All arrows and loops are assigned complex-valued weights, which change in time. The network dynamic is the time-evolution of the weights. The magnitude of the weights is interpreted as the capacity of their respective channels.
The weights are allowed to take on complex values in order to incorporate interference effects, which may be synergetic (amplifying) or uncooperative (suppressing), depending on the temporal distribution of phases.

A network dynamics is then encoded by a time-dependent complex matrix $F(t)$, so that $F_{ij}(t)$ is the weight of the arrow from node $i$ to node $j$, and $F_{jj}(t)$ is the weight of the loop at node $j$. Although many particular scenarios could be considered, the LOE dynamic (\ref{LOE_general}) has the advantage of being relatively simple. The choice of path $z(t)$ and the choice of $\Lambda$ determine the character of dynamics, see Fig. \ref{Fig_Sing_Values}. Note also that this type of dynamics factors to subalgebras, e.g. one may consider upper-triangular matrix solutions, etc.


In our interpretation, the deterministic LOE alone describes network evolution in the absence of external stimuli. We extended the model, introducing the multimodal RandLOE, in order to study evolution of networks stimulated by random exogenous factors. Samples of the resulting stochastic processes are given (via the singular values) in Fig. \ref{Fig_RandLOE_KL}. Naturally, numerical simulations take into account only finitely many terms of the (theoretically defined) infinite series. We have not investigated the problem of convergence, as it presents a formidable challenge that deviates from the core goal of this work,  cf. Remark \ref{Remark_conv}. However, we are satisfied that the model is computationally stable.

\section*{Acknowledgments}  We acknowledge helpful discussions  while this work was underway with Professors  Shahedul Khan, Raymond Spiteri, and FangXian Wu. We also acknowledge partial support of the BIT-UofS FLAGSHIP (a framework for scientific collaboration between the Beijing Institute of Technology and the University of Saskatchewan). 




\theendnotes

\end{document}